\documentclass[11pt]{llncs}
\usepackage{graphicx,cite,hyperref}
\usepackage{amsmath,amsfonts,amssymb}
\usepackage[margin=1in]{geometry}
\usepackage{microtype}
\usepackage[T1]{fontenc}

\newif{\ifFull}
\Fulltrue

\newtheorem{observation}{Observation}

\let\doendproof\endproof
\renewcommand\endproof{~\hfill$\qed$\doendproof}

\DeclareMathOperator{\hull}{\mathsf{hull}}
\DeclareMathOperator{\lowersets}{\mathsf{down}}
\DeclareMathOperator{\project}{\mathsf{project}}
\DeclareMathOperator{\extremes}{\mathsf{extremes}}
\DeclareMathOperator{\subproblem}{\mathsf{subproblem}}
\DeclareMathOperator{\freeset}{\mathsf{free}}
\DeclareMathOperator{\below}{\mathsf{below}}
\DeclareMathOperator{\powerset}{\mathsf{powerset}}
\DeclareMathOperator{\polygon}{\mathsf{polygon}}

\title{The Parametric Closure Problem}

\author{David Eppstein}
\institute{Computer Science Department, Univ. of California, Irvine, USA}

\begin{document}
\maketitle

\begin{abstract}
We define the \emph{parametric closure problem}, in which the input is a partially ordered set whose elements have linearly varying weights and the goal is to compute the sequence of minimum-weight downsets of the partial order as the weights vary. We give polynomial time solutions to many important special cases of this problem including semiorders, reachability orders of bounded-treewidth graphs, partial orders of bounded width, and series-parallel partial orders. Our result for series-parallel orders provides a significant generalization of a previous result of Carlson and Eppstein on bicriterion subtree problems.
\end{abstract}

\section{Introduction}
\emph{Parametric optimization} problems are a variation on classical combinatorial optimization problems such as shortest paths  or minimum spanning trees, in which the input weights are not fixed numbers, but vary as functions of a parameter. Different parameter settings will give different weights and different optimal solutions; the goal is to list these solutions and the intervals of parameter values within which they are optimal. As a simple example, consider maintaining the minimum of $n$ input values, which change as a parameter controlling these values changes. This parametric minimum problem can be formalized as the problem of constructing the \emph{lower envelope} of  functions that map the parameter value to each input value.
When these are linear functions this is the problem of constructing the lower envelope of lines,
equivalent by projective duality to a planar convex hull~\cite{Mat-LDG-LE-02}; the lower envelope has linear complexity and can be constructed in time $O(n\log n)$.

As well as the obvious applications of parametric optimization to real-world problems with time-varying but predictable data (such as rush-hour route planning), parametric optimization problems have another class of applications, to \emph{bicriterion optimization}. In bicriterion problems, each input element has two numbers associated with it. A solution value is obtained by summing the first number for each selected element, separately summing the second number for each selected element, and evaluating a nonlinear combination of these two sums.
 For instance, each element's two numbers might be interpreted as the $x$ and $y$ coordinates of a point in the plane associated with the element, and we might wish to find the solution
 whose summed $x$ and $y$ coordinates give a point as close to the origin as possible.
 The two numbers might represent  the mean and variance of a normal distribution, and we might wish to optimize some function of the summed distribution. The two numbers on each element might represent an initial investment cost and expected profit of a business opportunity, 
and we might wish to find a feasible combination of opportunities that maximizes the return on investment. Or, the two numbers might represent the cost and log-likelihood of failure of a communications link, and we might wish to find a low-cost communications network with high probability of success. Many natural bicriterion optimization problems of this type can be expressed as finding the maximum of a \emph{quasiconvex function} of the two sums (a function whose lower level sets are convex sets) or equivalently as finding the minimum of a quasiconcave function of the two sums. When this is the case, the optimal solution can always be obtained as one of the solutions of a parametric problem on only one variable, defined by re-interpreting the two numbers associated with each input element as the slope and $y$-intercept of a linear function that gives the weight of that element (a single number) as a function of the parameter~\cite{Kat-IEICE-92,CarEpp-SWAT-06}. In this way, any algorithm for solving a parametric optimization problem can also be used to solve bicriterion versions of the same type of optimization problem. Even though the bicriterion problem might combine its two numbers in a nonlinear way, the corresponding parametric problem uses linearly varying edge weights.

In this paper we formulate and provide the first study of the \emph{parametric closure problem}, the natural parametric variant of a classical optimization problem, the \emph{maximum closure problem}~\cite{Pic-MS-76,Hoc-MS-04}.

\subsection{Closures and parametric closures}
A closure in a directed graph is a subset of vertices such that all edges from a vertex in the subset go to another vertex in the subset; the maximum closure problem is the problem of finding the highest-weight closure in a vertex-weighted graph. Equivalently we seek the highest-weight \emph{downset} of a weighted partial order, where a downset is a subset of the elements of a partial order such that if $x<y$ in the order, and $y$ belongs to the subset, then $x$ also belongs to the subset. A partial order can be converted to an equivalent directed graph by considering each element of the order as a vertex, with an edge from $y$ to $x$ whenever $x<y$ in the order (note the reversed edge direction from the more usual conventions). In the other direction, given a directed graph, one may obtain an equivalent partial order on the strongly connected components of the graph, where component $x$ is less than component $y$ in the partial order whenever there is a path from a vertex in $y$ to a vertex in $x$ in the graph.

One of the classical applications of this problem involves open pit mining~\cite{LerGro-TCIMM-65}, where the vertices of a directed graph represent blocks of ore or of covering material that must be removed to reach the ore, edges represent ordering constraints on the removal of these blocks, and the weight of each vertex represents the net profit or loss to be made by removing that block. Other applications of this problem include military attack planning~\cite{Orl-NRLQ-87}, freight depot placement~\cite{Bal-MS-70,Rhy-MS-70}, scheduling with precedence constraints~\cite{Law-ADM-78,ChaEdm-Ord-85}, image segmentation~\cite{GibHanSon-ISAAC-11,AhmChoGib-WADS-13}, stable marriage with maximum satisfaction~\cite{IrvLeaGus-JACM-87}, and treemap construction in information visualization~\cite{BucEppLof-WADS-11}. Maximum closures can be found in polynomial time by a reduction to maximum flow~\cite{Pic-MS-76,Hoc-Nw-01} or by direct algorithms~\cite{FaaKimSch-MS-90}.

\begin{figure}[t]
\centering\includegraphics[width=0.75\textwidth]{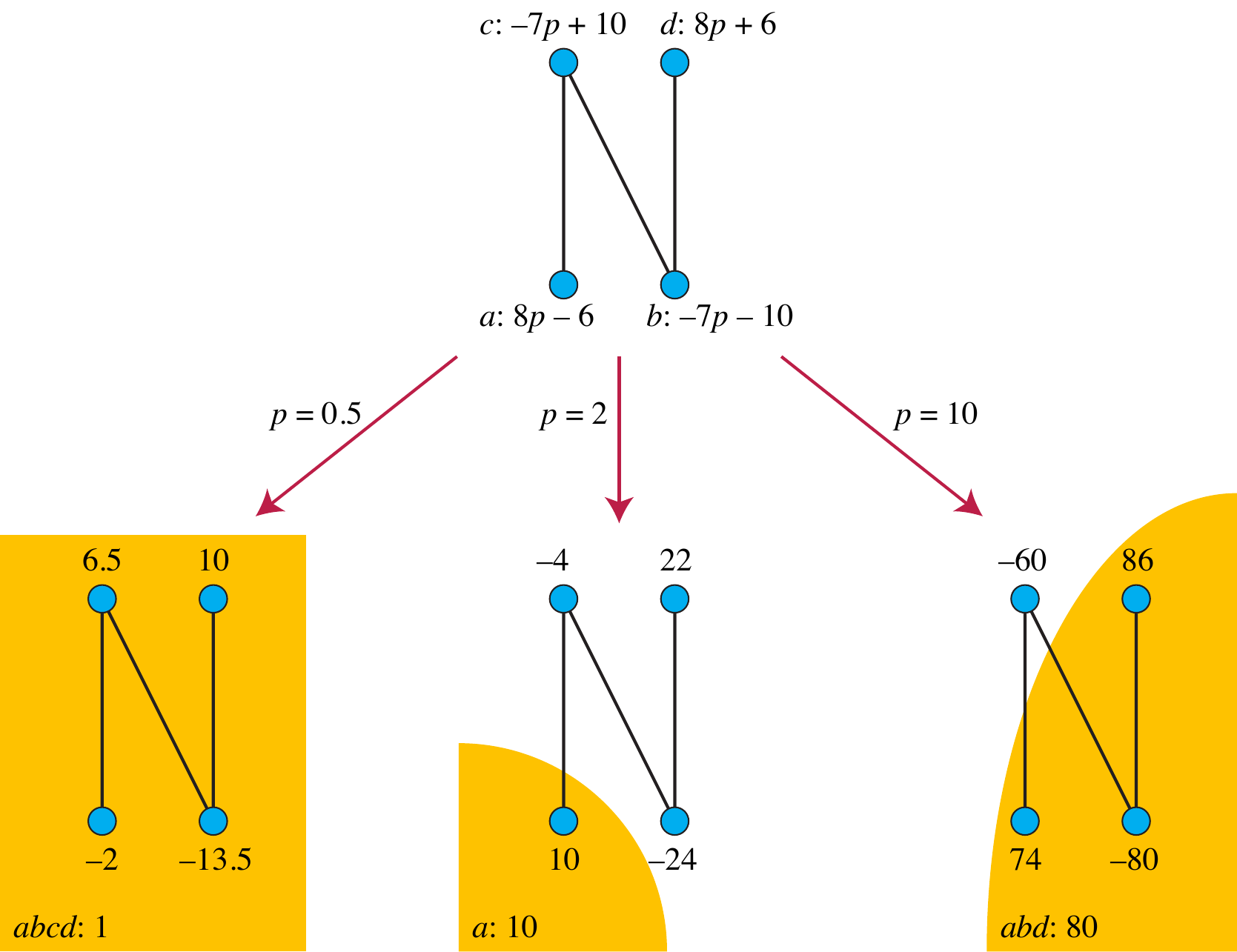}
\caption{In a partially ordered set with weights that vary linearly as a function of the parameter~$p$, different choices of $p$ lead to different maximum-weight downsets.}
\label{fig:N2}
\end{figure}

In the parametric closure problem, we assign weights to the vertices of a directed graph (or the elements of a partial order) that vary linearly as functions of a parameter, and we seek the closures (or downsets) that have maximum weight for each possible value of the parameter (\autoref{fig:N2}). For instance, in the open pit mining problem, the profit or loss of a block of ore will likely vary as a function of the current price of the refined commodity produced from the ore. So, a parametric version of the open pit mining problem can determine a range of optimal mining strategies, depending on how future commodity prices vary.  As described above, an algorithm for parametric closures can also solve bicriterion closure problems of maximizing a quasiconvex function (or minimizing a quasiconcave function) of two sums of values. Although we have not been able to resolve the complexity of the parametric closure problem in the general case, we prove near-linear or polynomial complexity for several important special cases of this problem.

\subsection{Related work}

Previous work on parametric optimization has considered parametric versions of two standard network optimization problems, the minimum spanning tree problem and the shortest path problem. The parametric minimum spanning tree problem (with linear edge weights) has polynomially many solutions that can be constructed in polynomial time~\cite{Epp-DCG-98,FerSluEpp-NJC-96,AgaEppGui-FOCS-98}. In contrast, the parametric shortest path problem is not polynomial, at least if the output must be represented as an explicit list of paths: it has a number of solutions and running time that are exponential in $\log^2 n$ on $n$-vertex graphs~\cite{Car-TR-84}.

We do not know of previous work on the general parametric closure problem, but two previous papers can be seen in retrospect as solving special cases:
\begin{itemize}
\item Lawler~\cite{Law-ADM-78} studied scheduling to minimize weighted completion time with precedence constraints. He sought the closure that maximizes the ratio $x/y$ of the priority $x$ and processing time $y$ of a job or set of jobs, and used this closure to decompose instances of this problem into smaller subproblems. Instead of using the reduction from bicriterion to parametric problems, Lawler showed that the optimal closure can be found in polynomial time by a binary search where each step involves the solution of a weighted closure problem. Replacing the binary search by Megiddo's parametric search~\cite{Meg-JACM-83} would make this algorithm strongly polynomial. Both search methods depend on the specific properties of the ratio function, however, and cannot be extended to other bicriterion problems.
\item
Carlson and Eppstein~\cite{CarEpp-SWAT-06} consider bicriterion versions of the problem of finding the best subtree (containing the root) of a given rooted tree with weighted edges. As they show, many such problems can be solved in time $O(n\log n)$. Although Carlson and Eppstein did not formulate their problem as a closure problem,
the root-containing subtrees can be seen as closures for a directed version of the tree in which each edge is directed towards the root. The reachability ordering on this directed tree is an example of a series-parallel partial order, and we greatly generalize the results of Carlson and Eppstein in our new results on parametric closures for arbitrary series-parallel partial orders.
\end{itemize}

\begin{figure}[t]
\centering\includegraphics[width=0.9\textwidth]{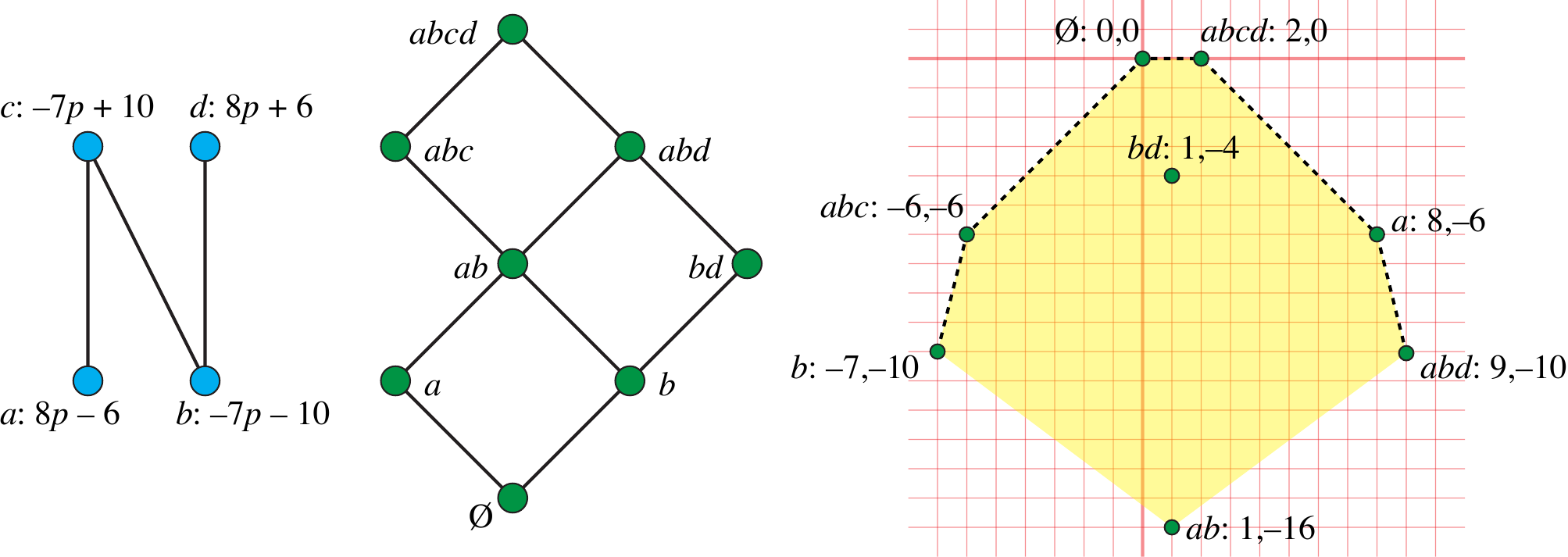}
\caption{An instance of the parametric closure problem. Left: The Hasse diagram of a partially ordered set $N$ of four elements, each with a weight that varies linearly with a parameter $p$. Center: The distributive lattice $\lowersets(N)$ of downsets of $N$. Right: The point set $\project(\lowersets(N))$ and its convex hull. The upper hull (dashed) gives in left-to-right order the sequence of six distinct maximum-weight  closures as the parameter $p$ varies continuously from $-\infty$ to $+\infty$.}
\label{fig:N}
\end{figure}

\subsection{Parametric optimization as an implicit convex hull problem}
Parametric optimization problems can be formulated dually, as problems of computing convex hulls of implicitly defined two-dimensional point sets (\autoref{fig:N}).
Suppose we are given a parametric optimization problem in which weight of element~$i$ is a linear function $a_i\lambda+b_i$ of a parameter $\lambda$, and in which the weight of a candidate solution $S$ (a subset of elements, constrained by the specific optimization problem in question) is the sum of these functions. Then the solution value is also a linear function, whose coefficients are the sums of the element coefficients:
\[
\sum_{i\in S} a_i\lambda + b_i = \left(\sum_{i\in S} a_i\right)\lambda+\left(\sum_{i\in S} b_i\right).
\]
Instead of interpreting the numbers $a_i$ and $b_i$ as coefficients of linear functions, we may re-interpret the same two numbers as the $x$ and $y$ coordinates (respectively) of points in the Euclidean plane.
In this way any family $\mathcal F$ of candidate solutions determines a planar point set, in which each set in $\mathcal F$ corresponds to the point given by the sum of its elements' coefficients. We call this point set $\project({\mathcal F})$, because the sets in $\mathcal F$ can be thought of as vertices of a hypercube $Q_n= \{0,1\}^n$ whose dimension is the number of input elements, and $\project$ determines a linear projection from these vertices to the Euclidean plane.

Let $\hull(\project({\mathcal F}))$ denote the convex hull of this projected planar point set.
Then for each parameter value the set in $\mathcal F$ minimizing or maximizing the parameterized weight corresponds by projective duality to a vertex of the hull, and the same is true for the maximizer of any quasiconvex function of the two sums of coefficients $a_i$ and $b_i$. Thus, parametric optimization can be reformulated as the problem of constructing this convex hull, and bicriterion optimization can be solved by choosing the best hull vertex.

\subsection{New results}
For an arbitrary partially ordered set $P$, define $\lowersets(P)$ to be the family of downsets of $P$. Let $P$ be parametrically weighted, so that $\project(\lowersets(P))$ is defined. As a convenient abbreviation, we define $\polygon(P)=\hull(\project(\lowersets(P)))$. This convex polygon represents the solution to the parametric closure problem for the given weights on $P$.

We consider the following classes of partially ordered sets. For each partial order $P$ in one of these classes,
we prove polynomial bounds on the complexity of $\polygon(P)$ and on the time for constructing $\polygon(P)$. These results imply the same time bounds for parametric optimization over~$P$ and for maximizing a quasiconvex function (bicriterion optimization) over $P$.

\begin{description}
\item[Semiorders.] This class of partial orders was introduced to model human preferences~\cite{Luc-Em-56} in which each element can be associated with a numerical value, pairs of elements whose values are within a fixed margin of error are incomparable, and farther-apart pairs are ordered by their numerical values. They are equivalently the interval orders with intervals of unit length, or the proper interval orders (interval orders in which no interval contains another).
For such orderings, we give in \autoref{sec:semiorder} a bound of $O(n\log n)$ on the complexity of $\polygon(P)$ and we show that it can be constructed in time $O(n\log^2 n)$ using an algorithm based on the quadtree data structure.

\item[Series-parallel partial orders.] These are orders formed recursively from smaller orders of the same type by two operations: series compositions (in which all elements from one order are placed earlier in the combined ordering than all elements of the other order) and parallel compositions (in which pairs of one element from each ordering are incomparable). These orderings have been applied for instance in scheduling applications by Lawler~\cite{Law-ADM-78}. For such orderings, the sets of the form $\polygon(P)$ have a corresponding recursive construction by two operations: the convex hull of the union of two convex polygons, and the Minkowski sum of two convex polygons. It follows that $\polygon(P)$ has complexity $O(n)$. This construction does not immediately lead to a fast construction algorithm, but in \autoref{sec:series-parallel} we adapt a splay tree data structure~\cite{SleTar-JACM-85} to construct $\polygon(P)$ in time $O(n\log n)$. Our previous results for optimal subtrees~\cite{CarEpp-SWAT-06}
follow as a special case of this result.

\item[Bounded treewidth.] Suppose that partial order $P$ has $n$ elements and its transitive reduction (the \emph{covering graph} of $P$) forms a directed acyclic graph whose undirected version has treewidth~$t$. (For prior work on treewidth of partial orders, see~\cite{JorMicMil-13}; for prior work on parametric optimization on graphs of bounded treewidth, see~\cite{FerSlu-Algs-97}.) Then we show in \autoref{sec:treewidth} that $\polygon(P)$ has polynomially many vertices, with exponent $t+2$, and that it can be constructed in polynomial time.

\item[Incidence posets.] The incidence poset of a graph $G$ has the vertices and edges as elements, with an order relation $x\le y$ whenever $x$ is an endpoint of $y$. One of the initial applications for the closure problem concerned the design of freight delivery systems in which a certain profit could be expected from each of a set of point-to-point routes in the system, but at the cost of setting up depots at each endpoint of the routes~\cite{Bal-MS-70,Rhy-MS-70}; this can be modeled with an incidence poset for a graph with a vertex at each depot location and an edge for each potential route. Since the profits and costs have different timeframes, it is reasonable to combine them in a nonlinear way, giving a bicriterion closure problem. The transitive reduction of an incidence poset is a subdivision of~$G$ with the same treewidth, so our technique for partial orders of bounded treewidth also applies to incidence posets of graphs of bounded treewidth.

\item[Fences, generalized fences, and polytrees.] The fences or zigzag posets are partial orders whose transitive reduction is a path with alternating edge orientations. A \emph{generalized fence} may be either an oriented path (an up-down poset)~\cite{Mun-SJDM-06} or an oriented tree (polytree)~\cite{Rus-Ord-89}. Polynomial bounds on the complexity and construction time for $\polygon(P)$ for all of these classes of partial orders follow from the result for treewidth~$t=1$. However, in these cases we simplify the bounded-treewidth construction and provide tighter bounds on these quantities (\autoref{thm:generalized-fence} and \autoref{thm:polytree}).

\item[Bounded width.] The \emph{width} of a partial order is the maximum number of elements in an \emph{antichain}, a set of mutually-incomparable elements. Low-width partial orders arise, for instance, in the edit histories of version control repositories~\cite{BanDevEpp-ANALCO-14}.
The treewidth of a partial order is less than twice its width,\footnote{To obtain a path-decomposition of width $2w-1$ from a partial order of width $w$, consider any linear extension of the partial order,
partition the partial order into a downset and an upset at each position of the linear extension, and form the sequence of sets that are the union of the maxima of a downset and the minima of the complementary upset.} but partial orders of width~$w$ have $O(n^w)$ downsets, tighter than the bound that would be obtained by using treewidth. We use quadtrees to show more strongly in \autoref{sec:width} that in this case $\polygon(P)$ has $O(n^{w-1}+n\log n)$ vertices and can be constructed in time within a logarithmic factor of this bound.
\end{description}

We have been unable to obtain an example of a family of partial orders with a nonlinear lower bound on the complexity of $\polygon(P)$, nor have we been able to obtain a nontrivial upper bound on the hull complexity for unrestricted partial orders. Additionally, we have been unable to obtain polynomial bounds on the hull complexity of the above types of partial orders for more than one parameter (i.e., for weights of dimension higher than two).  We also do not know of any computational complexity bounds (such as NP-hardness) for the parametric closure problem for any class of partial orders in any finite dimension. We leave these problems open for future research.

\section{Preliminaries: Minkowski sums and hulls of unions}

Our results on the complexity of the convex polygons $\polygon(P)$ associated with a partial order hinge on decomposing these polygons recursively into combinations of simpler polygons. To do this, we use two natural geometric operations that combine pairs of convex polygons to produce more complex convex polygons.

\begin{figure}[t]
\centering\includegraphics[width=0.75\textwidth]{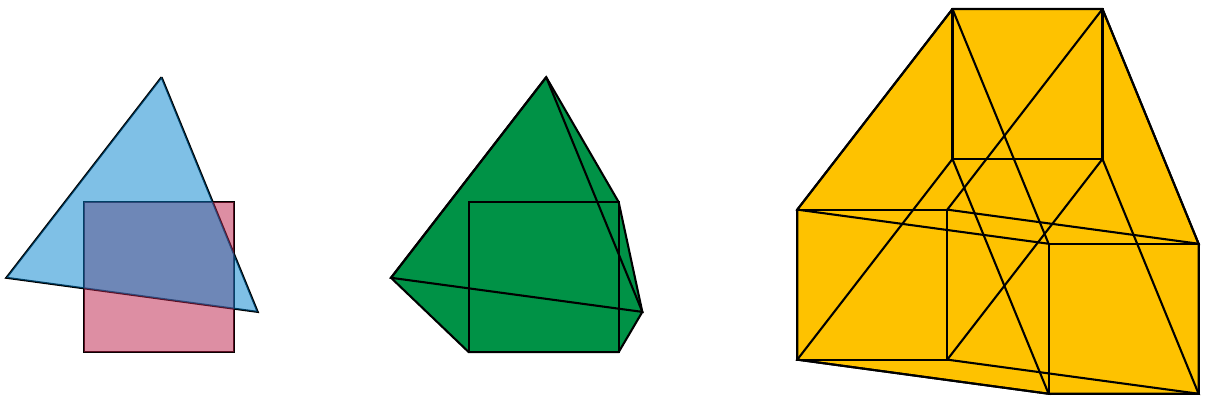}
\caption{Two convex polygons $P$ and $Q$ (left), the convex hull of their union $P\Cup Q$ (center),
and their Minkowski sum $P\oplus Q$ (right).}
\label{fig:polyops}
\end{figure}

\begin{definition}
We define a \emph{convex polygon} to be the convex hull of a nonempty finite set of points in the Euclidean plane. A \emph{vertex} of the polygon is a point that can be obtained as the intersection of the polygon with a closed halfplane, and an \emph{edge} of the polygon is a line segment that can be obtained as the intersection of the polygon with a closed halfplane. This definition requires us to include two degenerate special cases: we consider a single point to be a degenerate convex polygon with one vertex and no edges, and we consider a line segment to be a degenerate convex polygon with two vertices and one edge.
\end{definition}

\begin{definition}
For any two convex polygons $P$ and $Q$, let $P\oplus Q$ denote the Minkowski sum of $P$ and $Q$ (the set of points that are the vector sum of a point in $P$ and a point in $Q$), and let $P\Cup Q$ denote the convex hull of the union of $P$ and $Q$ (\autoref{fig:polyops}).
\end{definition}

\begin{lemma}[folklore]
\label{lem:combine-polygons}
If convex polygons $P$ and $Q$ have $p$ and $q$ vertices respectively, then $P\oplus Q$ and $P\Cup Q$ have at most $p+q$ vertices, and can be constructed from $P$ and $Q$ in time $O(p+q)$.
\end{lemma}

\begin{proof}
The complexity bound on $P\oplus Q$ follows from the fact that the set of edge orientations of $P\oplus Q$ is the union of the sets of edge orientations of $P$ and of $Q$. Therefore, $P\oplus Q$ has at most as many edges as the sum of the numbers of edges in $P$ and $Q$. The result follows from the fact that in any convex polygon, the numbers of vertices and edges are equal (except for the degenerate special cases, with one more vertex than edges). To compute $P\oplus Q$ from $P$ and $Q$, we may merge the two cyclically-ordered lists of edge orientations of the two polygons in linear time, and then use the sorted list to trace out the boundary of the combined polygon.

Similarly, the complexity bound on $P\Cup Q$ follows from the fact that its vertex set is a subset of the union of the vertices in $P$ and the vertices in $Q$. Therefore, its number of vertices is at most the sum of the numbers of vertices of $P$ and $Q$. To compute $P\Cup Q$ from $P$ and $Q$ in linear time, it is convenient to partition each convex hull into its \emph{lower hull} and \emph{upper hull}, two monotone polygonal chains, by splitting it at the leftmost and rightmost vertex of the hull.
We may sort the vertices of the upper hulls by doing a single merge in linear time, apply Graham scan to the sorted list, and do the same for the lower hulls.
\end{proof}

\begin{corollary}
\label{cor:formula-complexity}
Suppose that $P$ is a convex polygon, described as a formula that combines a set of $n$ points in the plane into a single polygon using a sequence of $\oplus$ and $\Cup$ operations. Suppose in addition that, when written as an expression tree, this formula has height~$h$. Then $P$ has at most $n$ vertices and it may be constructed from the formula in time $O(nh)$.
\end{corollary}

More complex data structures can reduce this time to $O(n\log n)$; see \autoref{sec:series-parallel}.

In higher dimensions, the convex hull of $n$ points and Minkowski sum of $n$ line segments both have polynomial complexity with an exponent that depends linearly on the dimension. However, we do not know of an analogous bound on the complexity of convex sets formed by mixing Minkowski sum and hull-of-union operations. If such a bound held, we could extend our results on parametric closures to the corresponding higher dimensional problems.

\section{Semiorders}
\label{sec:semiorder}

A \emph{semiorder} is a type of partial order defined by Luce~\cite{Luc-Em-56} to model human preferences. Each element of the order can be associated with a numerical value, which in the application to preference modeling is the \emph{utility} of that element to the person whose preferences are being modeled. For pairs of items whose utilities are sufficiently far from each other, the ordering of the two items in the semiorder is the same as the numerical ordering of their utilities. However, items whose utilities are within some (global) margin of error of each other are incomparable in the semiorder. More formally:

\begin{definition}
Let a collection of items $x_i$ be given, with numerical utilities $u_i$, together with a (global) threshold $\theta>0$. Then this information determines a partial order in which $x_i\le x_j$ whenever $u_i\le u_j-\theta$. The partial orders that can be obtained in this way are called the \emph{semiorders}. We will call $\theta$ the \emph{margin of error} of the semiorder.
\end{definition}

(We note that some authors use irreflexive binary relations, instead of partial orders, to define semiorders; this distinction will not be important for us.)

Similar concepts of comparisons of numerical values with margins of error give rise to semiorders in many other areas of science and statistics~\cite{PirVin-97}.
For efficient computations on semiorders we will assume that the utility values of each element are part of the input to an algorithm, and that the margin of error has been normalized to one. For instance, the semiorder $N$ of \autoref{fig:N} can be represented as a semiorder with utilities $2/3$, $0$, $2$, and $4/3$ for $a$, $b$, $c$, and $d$ respectively. With this information in hand, the comparison between any two elements can be determined in constant time. If only the ordering itself is given, numerical utility values can be constructed from it in time $O(n^2)$~\cite{Ave-Algs-92}.

The concept of a downset is particularly natural for a semiorder: it is a set of elements whose utility values could lie below a sharp numerical threshold, after perturbing each utility value by at most half the margin of error.
In this way, the closure problem (the problem of finding a maximum weight downset) can alternatively be interpreted as the problem of finding the maximum possible discrepancy of a one-dimensional weighted point set in which the location of each point is known imprecisely.
Thus, this problem is related to several other recent works on geometric computations with imprecise points (see, e.g.,~\cite{LofKre-CG-10,LofKre-Algo-10}).
Semiorders may have exponentially many downsets; for instance, if all items have utilities that are within one unit of each other, all sets are downsets. Nevertheless, as we show in this section, if $S$ is a semiorder, then the complexity of $\polygon(S)$ (the number of downsets that are optimal for some parameter value) is $O(n\log n)$.

\subsection{Mapping downsets to a grid}

If $S$ is any parametrically weighted semiorder, we may write the sorted order of the utility values of elements of $S$ as $u_0,u_1,\dots,u_{n-1}$ where $n=|S|$, and we may write the elements themselves (in the same order) as $x_0,x_1,\dots,x_{n-1}$. By scaling the utility values we may assume without loss of generality that the threshold $\theta$ used to define a semiorder from these values is set as $\theta=1$.
By padding $S$ with items that have a fixed zero weight and a utility that is smaller than that of the elements by more than the margin of error, we may additionally assume without loss of generality that $n$ is a power of two without changing the values of the parametric closure problem on $S$.

It is convenient to parameterize downsets by pairs of integers, as follows.

\begin{definition}
Let $L$ be an arbitrary downset in $\lowersets(S)$. Let $j$ be the largest index of an element $x_j$ of $L$. Let $i$ be the smallest index of an element $x_i$ such that $x_i$ does not belong to $L$ and $i<j$, or~$-1$ if no such element exists.
Define $\extremes(L)$ to be the pair of integers $(i+1,j)$.
\end{definition}

Thus, $\extremes$ maps the family $\lowersets(S)$ to the integer grid $[0,n-1]^2$. The mapping is many-to-one: potentially, many downsets may be mapped to each grid point. However, not every grid point is in the image of $\lowersets(S)$. In particular, a point $(i,j)$ with $i> j$ cannot be the image of a downset, because the element defining the first coordinate of $\extremes$ must have an index smaller than the element defining the second coordinate. Additionally, when $i>0$, a point $(i,j)$ with $u_{i-1}<u_j-1$ (i.e. with utility values that are beyond the margin of error for the semiorder) also cannot be the image of a downset, because in this case $x_{i-1}\le x_j$ in the semiorder, so every downset that includes $x_j$ also includes $x_{i-1}$. Thus, the image of $\extremes$ lies in an orthogonally convex subset of the grid, bounded below by its main diagonal and above by a monotone curve (\autoref{fig:semiorder-qt}).

\begin{figure}[t]
\centering
\includegraphics[width=0.9\textwidth]{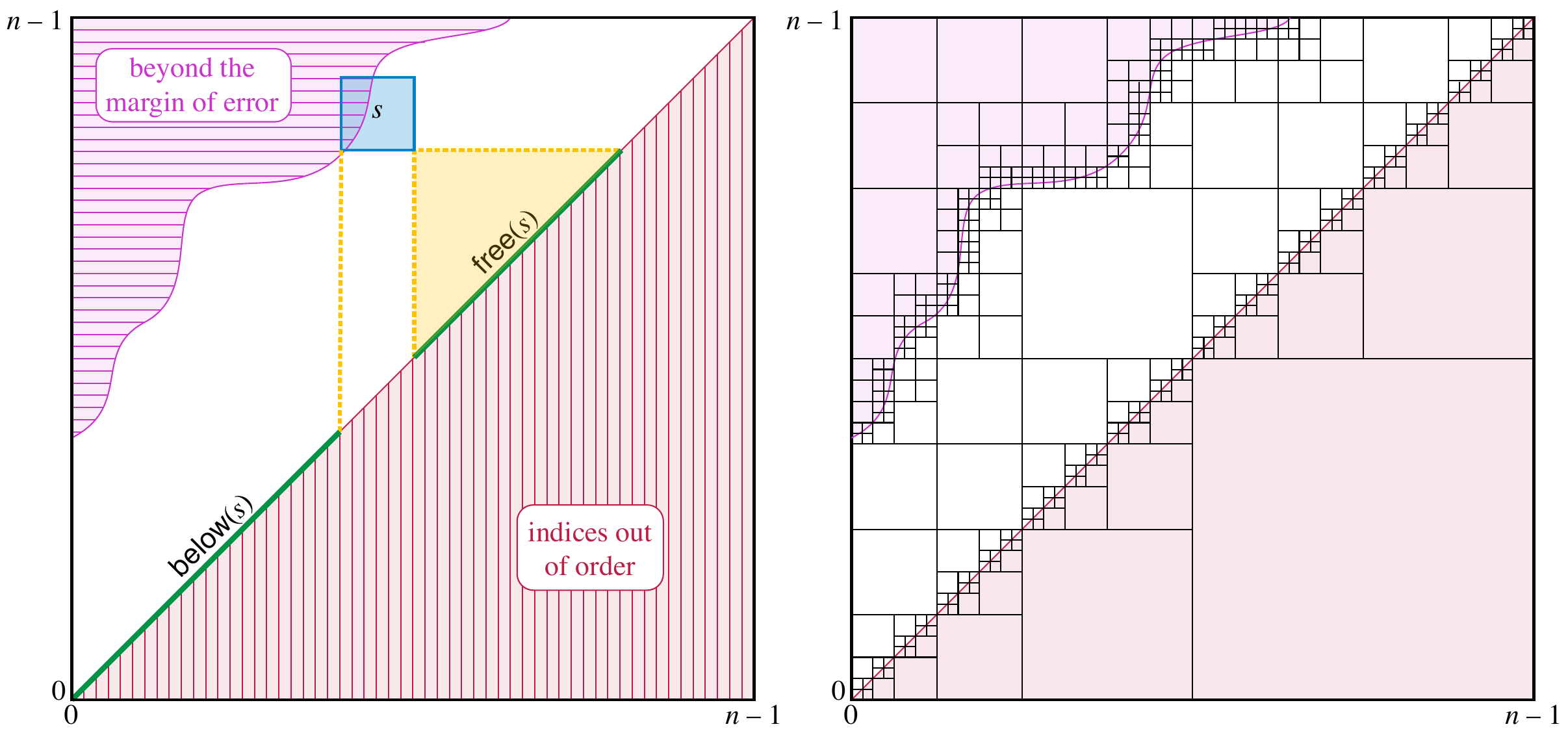}
\caption{The grid $[0,n-1]^2$, with the two regions that cannot be part of the image of $\extremes$. The left image shows a square subproblem $s$ and $\freeset(s)$; the right image shows the quadtree decomposition of the grid used to prove \autoref{thm:semiorder}.}
\label{fig:semiorder-qt}
\end{figure}

\begin{definition}
Let $s$ be any square subset of the integer grid $[0,n-1]^2$, and define $\subproblem(s)$ to be the partially-ordered subset of the semiorder $S$ consisting of the elements whose indices are among the rows and columns of $s$. Define $\freeset(s)$ to be the (unordered) set of elements of $S$ that do not belong to $\subproblem(s)$, but whose indices are between pairs of indices that belong to $\subproblem(s)$. Define $\below(s)$ to be the set of elements of $S$ whose indices are smaller than all elements of $\subproblem(s)$. (See \autoref{fig:semiorder-qt}, left, for an example.)
\end{definition}

\subsection{Decomposition of grid squares}
These definitions allow us to decompose the downsets that are mapped by $\extremes$ to the given square~$s$.

\begin{lemma}
\label{lem:characterize-semiorder-extremes}
Given a square $s$, suppose that the subfamily $\mathcal F$ of $\lowersets(S)$ that is mapped by $\extremes$ to $s$ is nonempty. Then each set in $\mathcal F$ is the disjoint union of three sets: a nonempty downset of $\subproblem(s)$, the set $\below(s)$, and an arbitrary subset of $\freeset(s)$.
\end{lemma}

\begin{proof}
Any downset of $S$ must remain a downset in any subset of $S$, and in particular its intersection with $\subproblem(s)$ must also be a downset. Additionally, it is not possible for a set in $\mathcal F$ to omit any member of $\below(s)$, nor to include any element outside $\subproblem(s)\cup\below(s)\cup\freeset(s)$, for then it would not be mapped into $s$ by $\extremes$. Similarly, the condition on the row index of the largest element of $\subproblem(s)$ must be met, or again $\extremes$ would map the given set outside of~$s$. Therefore, every set in $\mathcal F$ has the form described.
\end{proof}

Conversely, let a set be formed as the disjoint union of a nonempty downset of $\subproblem(s)$, $\below(s)$, and an arbitrary subset of $\freeset(s)$. Then this set is necessarily a downset, although it might not be mapped by $\extremes$ to $s$ (depending on the choice of the downset of $\subproblem(s)$ in the disjoint union).

We can decompose the convex hull of the projections of these downsets into a contribution from the subproblem of $s$ and another contribution from the free set of~$s$. The contribution from the free set has a simple structure based on Minkowski sums:

\begin{lemma}
\label{lem:powerset}
For any square $s$, let $\powerset(\freeset(s))$ be the family of all subsets of $\freeset(s)$.
Let weight function $w:S\to\mathbb{R}^2$ define a projection $\project$ from families of sets to point sets in $\mathbb{R}^2$.
Then $\project(\powerset(\freeset(s)))$ is the Minkowski sum of the sets $\{(0,0),w(x_i)\}$ for $x_i\in\freeset(s)$. Its convex hull is a centrally symmetric convex polygon  $\hull(\project(\powerset(\freeset(s))))$ (the Minkowski sum of the corresponding line segments) with at most $k=2|\freeset(s)|$ sides (fewer if some of these line segments are collinear), and can be constructed in time $O(k\log k)$.
\end{lemma}

\begin{proof}
The fact that $\project(\powerset(\freeset(s)))$ is a Minkowski sum of line segments can be proved by induction on the number of members of $\freeset(s)$: for any particular element $x$, the subfamilies of downsets that include or exclude $x$ project to translates of the same Minkowski sum, by the induction hypothesis, and the convex hull of the union of these two translates is the same thing as the Minkowski sum with one more line segment.
The bounds on the number of sides and construction time follow immediately from \autoref{cor:formula-complexity}, using the fact that the Minkowski sum of these segments can be expressed as a tree of binary Minkowski sums of logarithmic height.
\end{proof}

The $O(k\log k)$ time bound in \autoref{lem:powerset} can be reduced. The algorithm described, using a balanced tree of binary Minkowski sums, can be interpreted as performing a merge sort on the slopes of the line segments whose Minkowski sum is the desired hull by their slopes. If these segments were already sorted, the Minkowski sum could be constructed in linear time. And, in this case, it is possible to preprocess the input so that, for any contiguous subsequence of the elements, we can sort the corresponding line segments by their slopes quickly using integer sorting algorithms. However, we omit this speedup, as it adds complication to our overall algorithm without improving its total time.

The same proof used for \autoref{lem:powerset} also allows us to quickly compute $\polygon(\subproblem(s))$ for any square $s$ that is entirely within the margin of error, using the fact that the family of downsets of this subproblem is just its power set.

\begin{lemma}
\label{lem:unsubdivided-squares}
Let $s$ be a square of the grid defined by a given semiorder that is entirely within the margin of error of the semiorder, of side length $\ell$. Then $\polygon(\subproblem(s))$  is a polygon with $O(\ell)$ sides and can be computed in time $O(\ell\log\ell)$.
\end{lemma}

Computing $\polygon(\subproblem(s))$ for the remaining squares is trickier, but may be performed by decomposing the polygon into polygons of the same type for four smaller squares.

\begin{lemma}
\label{lem:semiorder-quadtree}
Let $s$ be a square in the grid $[0,n-1]^2$, containing an even number of grid points, and subdivide $s$ into four congruent smaller squares $s_i$ $(0\le i<4)$. Let $\polygon(\subproblem(s))$ have $c$ vertices, define $c_i$ in the same way for each $s_i$, and let $\ell$ be the side length of $s$. Then $c\le \sum c_i + O(\ell)$, and $\polygon(\subproblem(s))$ can be constructed from the corresponding hulls for the smaller squares in time $O(\sum c_i+\ell\log\ell)$.
\end{lemma}

\begin{proof}
For each smaller square $s_i$, define $H_i$ to be the polygon
\[
\polygon(\subproblem(s_i))\oplus \project(\powerset(\freeset(s_i)\cap\subproblem(s)))
\]
translated by adding all of the weights of elements in $\below(s_i)\cap\subproblem(s)$.

Then each vertex of this polygon $H_i$ represents a downset of $\subproblem(s)$.
By \autoref{lem:characterize-semiorder-extremes} (viewing $s_i$ as a subproblem of $\subproblem(s)$)  every downset of $\subproblem(s)$
that is mapped by $\extremes$ into $s_i$ is included in $H_i$.
(This polygon may also include some other downsets mapped into $s$ but not $s_i$, but this is not problematic.) Therefore, we can construct $\polygon(\subproblem(s))$ itself as
\[
\polygon(\subproblem(s))=H_0\Cup H_1\Cup H_2\Cup H_3.
\]
Thus, we have found a formula that expresses $\polygon(\subproblem(s)$ using a constant number of Minkowski sums and unions of hulls of the corresponding hulls for smaller squares,
together with free subproblems of total size $O(\ell)$. The result follows by using \autoref{lem:powerset} to construct the polygons for the free subproblems and \autoref{cor:formula-complexity} to combine the resulting polygons into a single polygon.
\end{proof}

\subsection{Semiorder algorithm}

Putting it all together, we can use the observations above to decompose $\polygon(S)$ into a combination of similarly-computed polygons for smaller grid squares, recursively decomposed into even smaller squares. This leads to an algorithm that performs this decomposition and then uses it in a bottom-up order to construct the polygons of larger subproblems from smaller ones, eventually producing the solution to the whole problem.

\begin{theorem}
\label{thm:semiorder}
If $S$ is a semiorder with $n$ elements $x_i$, specified with their utility values $u_i$ and a system of two-dimensional weights $w(x_i)$,
then $\polygon(S)$ has complexity $O(n\log n)$ and can be constructed in time $O(n\log^2 n)$.
\end{theorem}

\begin{proof}
We sort the utility values, pad $n$ to the next larger power of two if necessary and form a quadtree decomposition of the grid $[0,n-1]^2$ (as shown in \autoref{fig:semiorder-qt}, right). For each square $s$ of this quadtree, we associate a convex polygon (or empty set) $\polygon(\subproblem(s))$ computed according to the following cases:
\begin{itemize}
\item If $s$ is a subset of the grid points for which $i>j$, or for which $u_{i-1}<u_j-1$, then no downsets are mapped into $s$ by $\extremes$. We associate square $s$ with the empty set.
\item If $s$ is a subset of the grid points for which $i\le j$ and $u_{i-1}\ge u_j-1$, then every two elements of $\subproblem(s)$ are incomparable. In this case, we associate square $s$ with the polygon $\hull(\project(\powerset(\subproblem(s))))$ computed according to \autoref{lem:unsubdivided-squares}.
\item Otherwise, we split $s$ into four smaller squares. We construct the polygon associated with $s$ by using \autoref{lem:semiorder-quadtree} to combine the polygons associated with its children.
\end{itemize}

It follows by induction that the total complexity of the polygon constructed at any square $s$ of the quadtree is $O(\sum\ell_i)$, and the total time for constructing it is $O(\sum\ell_i\log n)$, where $\ell_i$ is the side length of the $i$th square of the quadtree and the sum ranges over all descendants of~$s$. Here, the $O(\ell_i\log n)$ contribution to the time bound includes not only the terms of the form $O(\ell_i\log\ell_i)$ appearing in \autoref{lem:powerset} and \autoref{lem:unsubdivided-squares}, but also the amount of time spent combining the polygon for this square with other polygons at $O(\log n)$ higher levels of the recursive subdivision.
 
As a base case for the induction, a square containing only a single grid point is associated with a subproblem with one element, with only one downset that maps to that grid point, and a degenerate convex polygon with a single vertex. The polygon constructed at the root of the quadtree is the desired output, and it follows that it has combinatorial complexity and time complexity of the same form, with a sum ranging over all quadtree squares.

The conditions $i>j$ and $u_{i-1}<u_j-1$ define two monotone curves through the grid, and we split a quadtree square only when it is crossed by one of these two curves. It follows that the squares of side length $\ell$ that are subdivided as part of this algorithm themselves form two monotone chains, and that the number of all squares of side length $\ell$ is $O(n/\ell)$. The results of the theorem follow by summing up the contributions to the polygon complexity and time complexity for the $O(\log n)$ different possible values of $\ell$.
\end{proof}

\section{Series-parallel partial orders}
\label{sec:series-parallel}

Series-parallel partial orders were considered in the context of a scheduling problem by Lawler~\cite{Law-ADM-78}, and include as a special case the tree orderings previously studied in our work on bicriterion optimization~\cite{CarEpp-SWAT-06}.

\subsection{Recursive decomposition}

Although it is possible to characterize the series-parallel partial orders as being the orders in which no four elements form the ``N'' of \autoref{fig:N}, it is more convenient for us to define them in terms of a natural recursive decomposition,
which we will take advantage of in our algorithms.

\begin{definition}
The \emph{series-parallel partial orders} are the partial orders that can be constructed from single-element partial orders by repeatedly applying the following two operations:

\begin{description}
\item[Series composition.] Given two series-parallel partial orders $P_1$ and $P_2$, form an order from their disjoint union in which every element of $P_1$ is less than every element of~$P_2$.
\item[Parallel composition.] Given two series-parallel partial orders $P_1$ and $P_2$, form an order from their disjoint union in which there are no order relations between $P_1$ and~$P_2$.
\end{description}
\end{definition}

\begin{figure}[t]
\centering\includegraphics[width=0.65\textwidth]{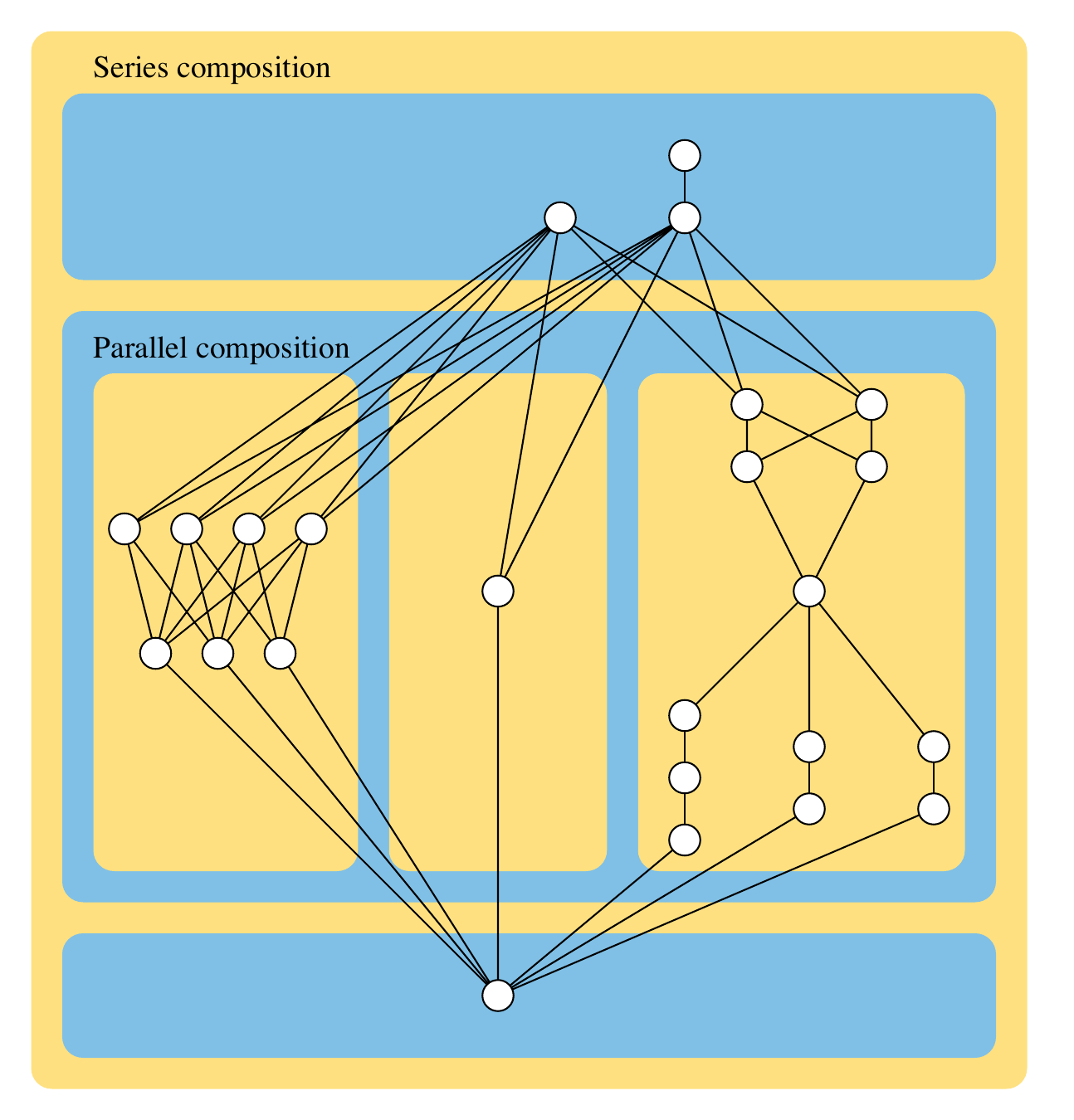}
\caption{A series-parallel partial order. Redrawn from a Wikipedia illustration by the author.}
\label{fig:SeriesParallel}
\end{figure}

See \autoref{fig:SeriesParallel} for an example.
These composition operations correspond naturally to the two geometric operations on convex polygons that we have already been using.

\begin{observation}
If $P$ is the series composition of $P_1$ and $P_2$, then $\polygon(P)$ is the convex hull of the union of $\polygon(P_1)$ and a translate (by the sum of the weights of the elements of $P_1$) of $\polygon(P_2)$. If $P$ is the parallel composition of $P_1$ and $P_2$, then $\polygon(P)$ is the Minkowski sum of $\polygon(P_1)$ and $\polygon(P_2)$.
\end{observation}

Recursively continuing this decomposition gives us a formula for $\polygon(P)$ in terms of the $\Cup$ and $\oplus$ operations. By \autoref{cor:formula-complexity} we immediately obtain:

\begin{corollary}
\label{cor:sp-complexity}
If $P$ is a series-parallel partial order with $n$ elements, then $\polygon(P)$ has at most $2n$ vertices.
\end{corollary}

However, the depth of the formula for $\polygon(P)$ may be linear, so using \autoref{cor:formula-complexity} to construct $\polygon(P)$ could  be inefficient. We now describe a faster algorithm. The key idea is to follow the same formula to build $\polygon(P)$, but to represent each intermediate result (a convex polygon) by a data structure that allows the $\Cup$ and $\oplus$ operations to be performed more quickly for pairs of polygons of unbalanced sizes.

\subsection{Data structure for fast unbalanced polygon merges}

Note that a Minkowski sum operation between a polygon of high complexity and a polygon of bounded complexity can change a constant fraction of the vertex coordinates, so to allow fast Minkowski sums our representation cannot store these coordinates explicitly.

\begin{lemma}
\label{lem:fast-combo}
It is possible to store convex polygons in a data structure such that destructively merging the representations of two polygons of $m$ and $n$ vertices respectively by a $\Cup$ or $\oplus$ operation (with $m<n$) can be performed in time $O(m\log((m+n)/m))$.
\end{lemma}

\begin{proof}
We store the lower and upper hulls separately in a binary search tree data structure, in which each node represents a vertex of the polygon, and the inorder traversal of the tree gives the left-to-right order of the vertices. The node at the root of the tree stores the Cartesian coordinates of its vertex; each non-root node stores the vector difference between its coordinates and its parents' coordinates. Additionally, each node stores the vector difference to its clockwise neighbor around the polygon boundary. In this way, we can traverse any path in this tree and (by adding the stored vector difference) determine the coordinates of any vertex encountered along the path. We may also perform a rotation in the tree, and update the stored vector differences, in constant time per rotation.

We will keep this tree balanced (in an amortized sense) by using the splay tree balancing strategy~\cite{SleTar-JACM-85}: whenever we follow a search path in the tree, we will immediately perform a splay operation that through a sequence of double rotations moves the endpoint of the path to the root of the tree. By the dynamic finger property for splay trees~\cite{ColMisSch-SJC-00,Col-SJC-00}, a sequence of $m$ accesses in sequential order into a splay tree of size $n$ will take time $O(m\log((m+n)/m))$.

To compute the hull of the union (the $\Cup$ operation) we insert each vertex of the smaller polygon (by number of vertices), in left-to-right order, into the larger polygon. To insert a vertex $v$, we search the larger polygon to find the edges with the same $x$-coordinate as $v$ and use these edges to check whether $v$ belongs to the lower hull, the upper hull, or neither. If it belongs to one of the two hulls, we search the larger polygon again to find its two neighbors on the hull. By performing a splay so that these neighbors are rotated to the root of the binary tree, and then cutting the tree at these points, we may remove the vertices between $v$ and its new neighbors from the tree without having to consider those vertices one-by-one. We then create a new node for $v$ and add its two neighbors as the left and right child.

To compute the Minkowski sum (the $\oplus$ operation) we must simply merge the two sequences of edges of the two polygons by their slopes. We search for each edge slope in the smaller polygon. When its position is found, we splay the vertex node at the split position to the root of the tree, and then split the tree into its left and right subtrees, each with a copy of the root node.
We translate all vertices on one side of the split by the vector difference for the inserted edge (by adding that vector only to the root of its tree), and rejoin the trees.
\end{proof}

The proof of the dynamic finger property for splay trees~\cite{ColMisSch-SJC-00,Col-SJC-00}
is very complicated, but the same lemma would follow using any other dynamic binary search data structure with the finger property that also allows for an additional operation of  removing large consecutive sequences of elements in time proportional to the logarithm of the length of the removed sequence. For instance, level-linked 2-3 trees have the finger property~\cite{BroTar-SICOMP-80} but the proven time bound for deletions in this structure is not of the form we need, so additional analysis would be needed to make them applicable to our problem.

\subsection{Series-parallel algorithm}

With this data structure in hand, we are ready to prove the main result of this section.

\begin{theorem}
If $P$ is a series-parallel partial order, represented by its series-parallel decomposition tree, then $\polygon(P)$ has complexity $O(n)$ and may be constructed in time $O(n\log n)$.
\end{theorem}

\begin{proof}
We follow the formula for constructing $\polygon(P)$ by $\Cup$ and $\oplus$ operations, using the data structure of \autoref{lem:fast-combo}. We charge each merge operation to the partial order elements on the smaller side of each merge. If a partial order element belongs to subproblems of sizes $n_0=1$, $n_1$, $\dots$, $n_h=n$ where $h$ is the height of the element, then the time charged to it is $\sum_i O(\log(n_i/n_{i-1}))=O(\log\prod_i(n_i/n_{i-1}))=O(\log n)$.
\end{proof}

\section{Bounded tree-width}
\label{sec:treewidth}

In this section, we consider partial orders whose underlying graphs have bounded treewidth.
The underlying graph is an undirected graph, obtained by forgetting the edge directions of the covering graph of the partial order. If the partial order is defined by reachability on directed acyclic graphs, the underlying graph of this reachability relation is obtained by forgetting the edge directions of the \emph{transitive reduction}, a minimal directed acyclic graph with the same reachability relation as the given one.

\subsection{Fences and polytrees}

\begin{figure}[t]
\centering\includegraphics[width=0.9\textwidth]{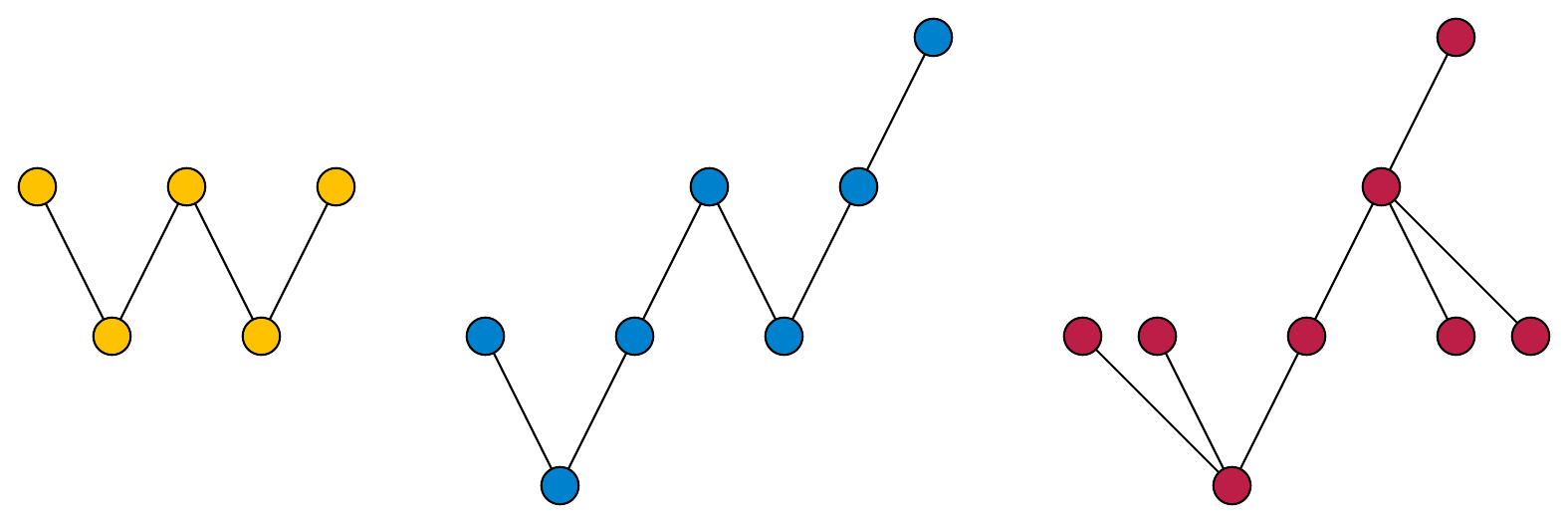}
\caption{A fence (left), generalized fence (center), and polytree (right)}
\label{fig:fences}
\end{figure}

As a warm-up, we first consider the case where the underlying undirected graph has treewidth one: that is, it is a tree or a path.
We can bound both the number of solutions to the parametric closure problem and the time for listing all the solutions using a divide-and-conquer approach that we will later generalize to larger values of the treewidth.

\begin{definition}
A \emph{fence} or \emph{zigzag poset} is a partially ordered set whose transitive reduction is a path with alternating edge orientations. A \emph{polytree} is a directed graph formed by orienting the edges of an undirected tree.
An \emph{oriented path} or \emph{up-down poset} is a directed graph formed by orienting the edges of an undirected simple path.
A \emph{generalized fence} is a partially ordered set whose transitive reduction is an oriented path.
(See \autoref{fig:fences}.)
\end{definition}

\autoref{fig:N} depicts an example of a fence.
Our definition of a generalized fence follows Munarini~\cite{Mun-SJDM-06} and should be distinguished from other authors who define generalized fences as the reachability orders of polytrees~\cite{Rus-Ord-89}.

\begin{theorem}
\label{thm:generalized-fence}
Let $P$ be a generalized fence.
Then $\polygon(P)$ has complexity $O(n^2)$ and can be constructed in time $O(n^2\log n)$.
\end{theorem}

\begin{proof}
We let $F(n)$ denote the maximum number of vertices of $\polygon(P)$ where $P$ ranges over all $n$-element generalized fences. To bound $F(n)$, let $P$ be a generalized fence of maximum complexity, and let $v$ be the middle element of the oriented path from which $P$ is defined.
Let $L^-$ be the generalized fence determined by the half of the path to the left of $P$, after removing all elements below $v$ in this half of the path, and let $L^+$ be the generalized fence determined by the left half-path after removing all elements above $v$. Define $R^-$ and $R^+$ in the same way for the right half of the path. Then the downsets of $P$ that contain $v$ correspond one-for-one with downsets of $L^-\cup R^-$, where the correspondence is obtained by adding back $v$ and all elements below $v$. The downsets of $P$ that do not contain $v$ are exactly the downsets of $L^+\cup R^+$. Therefore, we have the formula
\[
\polygon(P) = (\polygon(L^-)\oplus\polygon(R^-))\Cup (\polygon(L^-)\oplus\polygon(R^-)),
\]
with a suitable translation applied before performing the $\Cup$ operation.
Since all four of $L^-$, $L^+$, $R^-$, and $R^+$ have at most $\lfloor n/2\rfloor$ elements,
we can use this formula and \autoref{lem:combine-polygons} to derive a recurrence for the complexity $F(n)$:
\[
F(n)\le 4F(\lfloor n/2\rfloor),
\]
with the base case $F(0)=1$.
This recurrence has a standard form familiar from divide-and-conquer algorithms,
whose solution is $F(n)=O(n^2)$. Using the same formula for $\polygon(P)$ together with \autoref{lem:combine-polygons} to construct the polygon gives another recurrence for the running time $T(n)$:
\[
T(n) \le 4T(\lfloor n/2\rfloor)+O(F(n))=4T(\lfloor n/2\rfloor)+O(n^2)=O(n^2\log n).
\]
\end{proof}

We may apply the same divide-and-conquer method to construct the solutions to the parametric closure problem for an oriented tree, but we will not always be able to obtain a perfect split into two subproblems of size $n/2$. As a consequence, the method becomes somewhat more complex although its asymptotic time bound will remain the same. To pick the splitting vertex $v$ we use the following separator theorem for trees, which appears already in the work of Jordan (1869)~\cite{Jor-JRAM-69}:

\begin{lemma}
\label{lem:tree-halver}
Any $n$-vertex undirected tree has a vertex $v$ such that the removal of $v$ splits the remaining tree into components of at most $n/2$ vertices each.
\end{lemma}

\begin{proof}
Let $v$ be a vertex of the tree $T$, chosen to minimize the number $e$ of vertices by which the largest component of $T\setminus v$ exceeds $n/2$, and suppose for a contradiction that $e$ is not zero. Let $C$ be the component of $T\setminus v$ with $n/2+e>n/2$ vertices, and let $w$ be the unique neighbor of $v$ in $C$. Then removing $w$ from $T$ produces components that are strict subsets of $C$, together with a component containing $v$ with $n/2-e<n/2$ vertices.
Therefore, all components of $T\setminus w$ have fewer than $n/2+e$ vertices, contradicting the choice of $v$ as the minimizer of $e$. This contradiction implies that $e$ must be zero and therefore that removing $v$ splits the remaining tree into components of at most $n/2$ vertices each.
\end{proof}

\begin{theorem}
\label{thm:polytree}
Let $P$ be the reachability order of a polytree $T$.
Then $\polygon(P)$ has complexity $O(n^2)$ and can be constructed in time $O(n^2\log n)$.
\end{theorem}

\begin{proof}
As in \autoref{thm:generalized-fence} we remove $v$ from $T$, and for each component $C_i$ of the resulting forest define partial orders $C_i^-$ and $C_i^+$ by removing the elements below or above $C_i$ respectively. And as in \autoref{thm:generalized-fence} we can use this decomposition to derive a formula
\[
\polygon(P)=(\polygon(C_1^-)\oplus\polygon(C_2^-)\oplus\cdots)\Cup (\polygon(C_1^+)\oplus\polygon(C_2^+)\oplus\cdots)
\]
By \autoref{lem:combine-polygons}, it follows that the total number of vertices of $\polygon(P)$ is at most the total number of leaf-level subproblems obtained by recursively performing this decomposition.
Because the size of the remaining subtrees goes down by a factor of two in each level of recursion, there are at most $\log_2 n$ levels of recursion. Each level of recursion also at most doubles the number of subproblems that each remaining tree vertex belongs to: a vertex $u$ that belongs to a component $C_i$ of the remaining tree vertices will also be included in at most two subproblems $C_i^-$ and $C_i^+$. Therefore, the contribution of any one tree vertex to the total complexity of $\polygon(P)$ is $O(n)$ and the total complexity of $\polygon(P)$ is $O(n^2)$.

We cannot use \autoref{lem:combine-polygons} directly to evaluate this formula, because that lemma allows only binary combinations. However, using \autoref{lem:fast-combo} instead allows any formula of $\oplus$ and $\Cup$ operations of size $s$ to be evaluated in time $O(s\log s)$. Applying this method to the recursive formula derived above gives us a total time of $O(n^2\log n)$.
\end{proof}

In the following sections, we apply the same method to more general graphs of bounded treewidth.

\subsection{Tree decompositions}

\begin{definition}
A \emph{tree-decomposition} of an undirected graph $G$ is a tree $T$ whose vertices are called \emph{bags},
each of which is associated with a set of vertices of $G$, with the following two properties:
\begin{itemize}
\item Each vertex of $G$ belongs to a set of bags that induces a connected subtree of~$T$, and
\item For each edge of $G$ there exists a bag containing both endpoints of the edge.
\end{itemize}
The \emph{width} of a tree-decomposition is one less than the maximum number of vertices in a bag, and the \emph{treewidth} of $G$ is the minimum width of a tree-decomposition of~$G$. By extension, we define the treewidth of a directed graph $D$ to equal the treewidth of the undirected graph obtained by forgetting the orientations of the edges of $D$.
\end{definition}

The treewidth, and an optimal tree-decomposition, can be found in an amount of time that is linear in the number of vertices of a given graph but exponential in its width~\cite{Bod-SJC-96}.

\begin{definition}
A tree-decomposition of width $t$ is \emph{minimal} if no two adjacent bags have a union with at most $t+1$ elements.
\end{definition}

Every graph of treewidth $t$ has a minimal decomposition of width $t$.
For, given a decomposition that is not minimal (in which two adjacent bags have a small union), the edge connecting those two bags could be contracted, and the bags merged, decreasing the number of bags in the decomposition without increasing the width of the decomposition. This property allows us to control the number of bags in the decomposition:

\begin{lemma}
A minimal tree-decomposition of width $t$ of an $n$-vertex graph has at most $n-t$ bags.
\end{lemma}

\begin{proof}
Consider the bags in the order given by a breadth-first traversal, starting from a bag of cardinality $t+1$. The first bag has $t+1$ vertices in it, and each successive bag in this traversal must have at least one new vertex not seen in previous bags, for otherwise it would be a subset of the bag that is its parent in the breadth-first traversal and the decomposition would not be minimal. Therefore, the number of vertices is at least equal to the number of bags plus $t$. Equivalently, the number of bags is at most $n-t$.
\end{proof}

\subsection{Partial orders of low treewidth}

\begin{lemma}
\label{lem:tw-formula}
Let $P$ be a partial order with $n$ elements whose underlying graph has treewidth~$t=O(1)$.
Then $\polygon(P)$ can be represented as a formula of $\Cup$ and $\oplus$ operations with 
total size $O(n^{t+2})$ and depth $O(\log n)$.
\end{lemma}

\begin{proof}
We fix a minimal tree-decomposition of width $t$, and we use \autoref{lem:tree-halver} to recursively partition this decomposition by repeatedly removing a bag $B$ whose removal splits the remaining part of the decomposition into connected components $C_i$ of size (number of bags) at most half the size at the previous level of the partition.

If $B$ has $b$ elements, then there are at most $2^b$ partitions of $B$ into two subsets $L$ and $U$, such that $L$ is a downset in $B$ (with $U$ as its complementary upset).
We will use these partitions to specify the elements of $B$ that should be part of an eventual downset ($L$) or should not be part of an eventual downset ($U$). Let $C_i$ be one of the connected components formed from the tree-decomposition by the removal of $B$, and $(L_j,U_j)$ be one of the partitions of $B$ into a downset $L=L_j$ and an upset $U=U_j$. Then in any downset of the whole problem consistent with the partition $(L_j,U_j)$, the elements of $C_i$ that are below at least one element of $L_j$ must also be included in the downset, and the elements of $C_i$ that are above at least one element of $U_j$ must also be included in the upper set. We define $C_i^j$ to be the subproblem of the closure problem consisting of the remaining elements: the ones that belong to bags of $C_i$ and are neither below a member of $L_j$ nor above a member of $U_j$.

With this notation, we can express $\polygon(P)$ as a formula
\[
\polygon(P)=\Cup_{j=1}^{2^b} \oplus_i \polygon(C_i^j),
\]
where each term in the outer $\Cup$ is translated appropriately according to the weights of the additional elements in or below $L_j$ that belong to the downsets for that term but not to the subproblems included in that term.

Each of the at most $\log_2 n$ levels of this recursive partition causes the number of subproblems that each vertex participates in to increase by a factor of at most $2^{t+1}$. Therefore, in the whole recursive partition, the total number of subproblems that a single vertex participates in is $O(n^{t+1})$. Since there are $n$ vertices, this partition gives rise to a formula with total complexity $O(n^{t+2})$.
\end{proof}

Applying \autoref{lem:fast-combo} to the formula of \autoref{lem:tw-formula} gives us the following result:

\begin{theorem}
\label{thm:treewidth}
Let $P$ be a partial order with $n$ elements whose underlying graph has treewidth~$t=O(1)$.
Then $\polygon(P)$ has complexity $O(n^{t+2})$ and can be constructed in time $O(n^{t+2}\log n)$.
\end{theorem}

We note that for $t=1$ this is larger by a linear factor than the bounds of \autoref{thm:generalized-fence} and \autoref{thm:polytree}. We do not know whether the dependence on $t$ in the exponent of the complexity bound is necessary, but unless it can be improved it acts as an obstacle to the existence fixed-parameter-tractable algorithms parameterized by treewidth for the construction of $\polygon(P)$.

\subsection{Incidence posets}

The partial orders defined by reachability on graphs of low treewidth include, in particular, the incidence posets of undirected graphs of bounded treewidth. This is because the incidence poset of an undirected graph $G$ is the same as the reachability poset on a graph obtained from $G$ by subdividing each undirected edge, and orienting each of the two resulting new edges outward from
the subdivision point (\autoref{fig:incidence}). This replacement cannot increase the treewidth, as the following folklore lemma shows.

\begin{figure}[t]
\centering\includegraphics[width=0.8\textwidth]{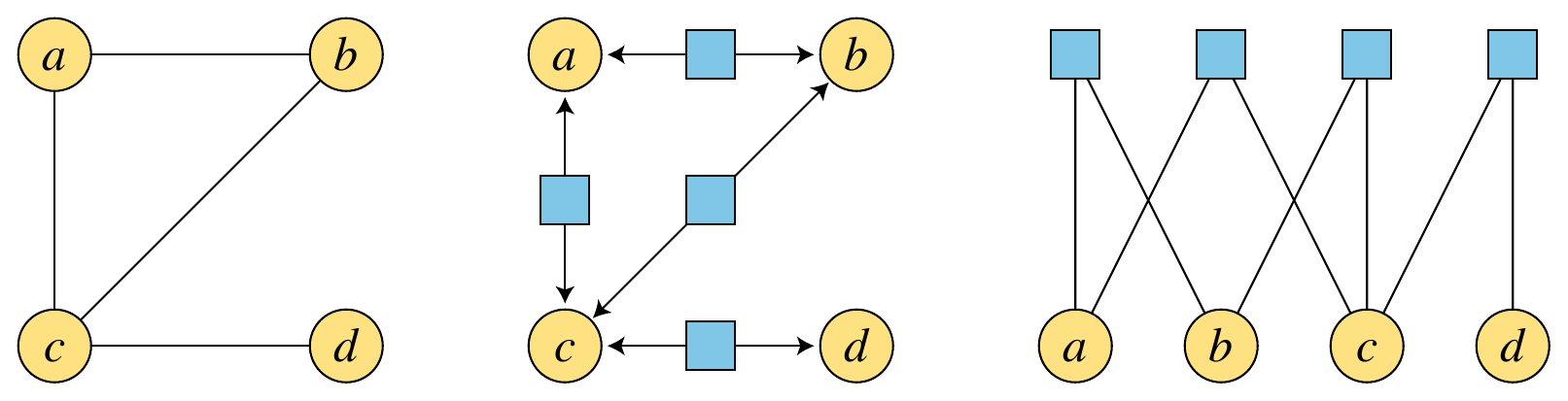}
\caption{An undirected graph (left), the directed graph obtained by replacing each undirected edge by two oppositely-ordered directed edges (center), and a partial order that is simultaneously the incidence poset of the undirected graph and the reachability poset of the directed graph (right).}
\label{fig:incidence}
\end{figure}

\begin{lemma}
Let $H$ be obtained from $G$ by replacing one or more edges of $G$ by paths.
Then the treewidth of $H$ is no higher than that of $G$.
\end{lemma}

\begin{proof}
We may assume that $H$ is obtained from $G$ by replacing only a single edge $e$ by a two-edge path, for any other instance of the lemma may be obtained by multiple replacements of this type.
If $G$ is a tree, then $H$ is also a tree, and both have treewidth $1$. Otherwise, the treewidth of $G$ is at least $2$, and an optimal tree-decomposition for $G$ must contain at least one bag $B$ that contains both endpoints of $e$. A tree-decomposition for $H$ of the same width may be obtained by attaching to $B$ another bag that  contains three vertices: the endpoints of $e$, and the subdivision point of the replacement path for $e$.
\end{proof}

One of the important applications of the closure problem, to transportation planning, uses the incidence posets of graphs whose edges represent a collection of point-to-point truck routes~\cite{Bal-MS-70,Rhy-MS-70}. Our methods can be applied to the parametric and bicriterion versions of this problem when its graph has small treewidth.

The following result is an immediate corollary of \autoref{thm:treewidth}.

\begin{corollary}
Let $P$ be the incidence poset of an $n$-vertex graph of treewidth~$w$.
Then $\polygon(P)$ has complexity $O(n^{w+2})$ and can be constructed in time $O(n^{w+2}\log n)$.
\end{corollary}

\section{Bounded width}
\label{sec:width}

The \emph{width} of a partially ordered set is the size of its largest antichain. Partially ordered sets of bounded width arise, for instance, in the version histories of a distributed version control repository controlled by a small set of developers (with the assumption that each developer maintains only a single branch of the version history). In this application, there may be many elements of the partially ordered set (versions of the repository) but the width may be bounded by the number of developers~\cite{BanDevEpp-ANALCO-14}. A downset in this application is a set of versions that could possibly describe the simultaneous states of all the developers at some past moment in the history of the repository.

The downsets of a partially ordered set correspond one-for-one with its independent sets: each downset is uniquely determined by an independent set, the set of maximal elements of the downset. Therefore, in a partially ordered set of width $w$ with $n$ elements there can be at most $O(n^w)$ downsets. More precisely, by Dilworth's theorem~\cite{Dil-AM-50}, every partial order of width $w$ can be partitioned into $w$ totally-ordered subsets (chains); the number of downsets can be at most the product of the numbers of downsets of these chains. This product takes its maximum value $(n/w+1)^w$ when the chains all have equal size.
As we show, however, there is an even tighter bound for the complexity of $\polygon(P)$.

We consider first the case $w=2$. The more general case of bounded width will follow by similar reasoning. Thus, let $P$ be a partially ordered set with $n$ elements and width two. By Dilworth's theorem, $P$ can be decomposed into two chains, and every downset $L$ can be described by a pair of integer coordinates $(x,y)$, where $x$ is the number of elements of $L$ that belong to the first chain and $y$ is the number of elements of $L$ in the second chain.

\begin{observation}
For a width-two partial order $P$ as described above, the set of pairs of coordinates $(x,y)$ describing the downsets of $P$ forms an orthogonally convex subset of the integer grid $[0,n]^2$, and the edges between points one unit apart in this set of grid points form the covering graph of the distributive lattice of downsets.
\end{observation}

\begin{figure}[t]
\centering\includegraphics[scale=0.3]{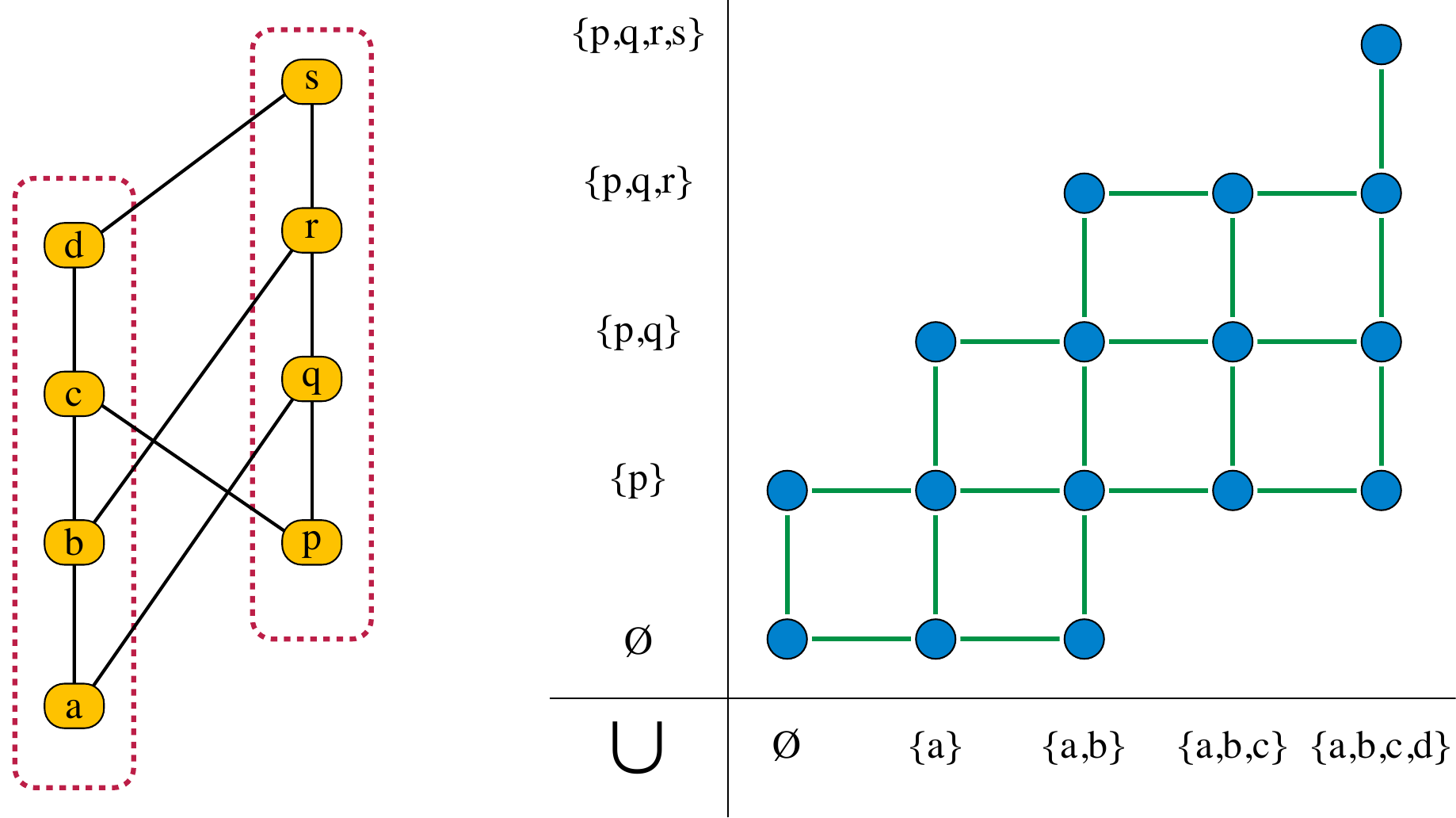}
\caption{Hasse diagram of a width-two partial order partitioned into two chains (left) and the subset of the integer lattice formed by its downsets (right)}
\label{fig:width2polyomino}
\end{figure}

\autoref{fig:width2polyomino} depicts an example.
We can use the following definition and observation to partition the parametric closure problem for~$P$ into smaller subproblems of the same type.

\begin{definition}
Let $R$ be any grid rectangle. Then $\subproblem(R)$ is the family of downsets of $P$ whose grid points lie in $R$.
\end{definition}

\begin{observation}
\label{obs:w2-partial}
For every grid rectangle~$R$, each set in $\subproblem(R)$ can be represented uniquely as a disjoint union $A\cup B$ where $A$ is the set of all elements of $P$ whose (single) coordinate value is below or to the left of $R$, and $B$ is a downset of the restriction of $P$ to the elements whose coordinate value is within~$R$.
\end{observation}

In this decomposition, the same set $A$ forms a subset of every set in $\subproblem(R)$. Thus, the observation shows that the contribution to the parametric closure problem from rectangle $R$ can be obtained by solving a smaller parametric closure problem, on the restriction of the partial order to the elements within~$R$, and then translating the results by a single vector determined from this fixed set $A$.

\begin{observation}
\label{obs:w2-full}
Suppose that every integer point of a grid rectangle $R$ corresponds to a downset of $P$.
Then the restriction of $P$ to the elements whose coordinate value is within~$R$ is a partial order in the form of the parallel composition of two chains, a special case of a series-parallel partial order. It follows from \autoref{cor:sp-complexity} that $\hull(\project(\subproblem(R)))$ has complexity proportional to the perimeter of~$R$.
\end{observation}

\begin{theorem}
If $P$ is a partial order of width~two, then $\polygon(P)$ has complexity $O(n\log n)$ and can be constructed from the covering relation of~$P$ in time $O(n\log^2 n)$.
\end{theorem}

\begin{proof}
From the covering relation, we can determine a partition into two chains and trace out the boundaries of the orthogonal grid polygon describing the downsets of~$P$, in linear time.
We use a quadtree to partition the downsets of $P$ into squares, each with a side length that is a power of two, such that each grid point within each square corresponds to one of the downsets of~$P$. The squares of each size form two monotone chains within the grid, so the total perimeter of all of these squares is~$O(n\log n)$.

We also use the one-dimensional projection of the same quadtree, to partition each of the two grid coordinates recursively into subintervals whose sizes are powers of two. For each subinterval~$I$, we compute the convex hull of the downsets of the corresponding chain whose grid points lie within~$I$, as the hull of the union of the two previously-constructed polygons for the two halves of~$I$. This computation takes time $O(n\log n)$ overall. 
We also compute a sequence of vectors, the sums of weights of each prefix of the coordinates, in time $O(n)$.

We then compute the polygon for each of the squares of the partition of the downsets of $P$.
By \autoref{obs:w2-partial}, each such polygon can be constructed by
translating the polygon for the grid coordinates within the square (the set $B$ of the observation)
by a translation vector determined from all the smaller grid coordinates (the set $A$ of the observation).
By \autoref{obs:w2-full}, the polygon for the set $B$ can be computed in time proportional to the perimeter of the square, as the Minkowski sum of the polygons for its two defining subintervals.
The translation vector for the set $A$ can be found as the sum of two vectors computed for the two corresponding prefixes of the two grid coordinates.

The overall solution we seek, $\polygon(P)$, can be computed by applying the $\Cup$ operation to combine the polygons computed for each square of the partition. As the total perimeter of these squares is $O(n\log n)$, and each generates a polygon of complexity proportional to its perimeter in time proportional to its perimeter, the total time to compute all of these polygons is $O(n\log n)$, and their total complexity is $O(n\log n)$. As the convex hull of $O(n\log n)$ points, the time to perform the final $\Cup$ operation to combine these polygons and compute $\polygon(P)$ is $O(n\log^2 n)$.
\end{proof}

For higher widths~$w$, the same idea works (using an octree in three dimensions, etc). The total complexity of $\polygon(P)$ is proportional to the sum of side lengths of octree squares, $O(n^{w-1})$, and the construction time is within a logarithmic factor of this bound.

\section{Conclusions}

We have introduced the parametric closure problem, and given polynomial complexity bounds and  algorithms for several important classes of partial orders. Bounds for the general problem, and nontrivial lower bounds on the problem, remain open for further research.

\section*{Acknowledgements}
This research was supported in part by NSF grant 1228639 and
ONR grant N00014-08-1-1015.

{\raggedright
\bibliographystyle{abuser}
\bibliography{closure}}

\end{document}